\documentclass[conference]{ieeeconf}  

\IEEEoverridecommandlockouts                              
\usepackage{float}
\usepackage[scaled=.92]{helvet}

\usepackage{amssymb}
\usepackage{amsmath}
\usepackage{bbold}
\usepackage{caption}
\usepackage{comment}
\usepackage[ruled]{algorithm2e}
\usepackage{float}
\usepackage{multicol}
\usepackage{color}
\usepackage{graphicx}
\newcounter{multifig}
\usepackage{thmtools,thm-restate}
\usepackage{multirow}
\usepackage{mathtools, cuted}

\usepackage{lipsum}

\usepackage{array}
\usepackage{makecell}

\usepackage[caption=false, font=footnotesize]{subfig}

\newtheorem{theorem}{Theorem}[section]
\newtheorem{lemma}[theorem]{Lemma}

\newtheorem{definition}[theorem]{Definition}
\newtheorem{remark}[theorem]{Remark}
\newtheorem{example}[theorem]{Example}

\title{\LARGE \bf Communication-efficient Distributed Multi-resource Allocation}

\author{Syed Eqbal Alam$^\ast$\thanks{$^\ast$Concordia Institute for Information Systems Engineering,
		Concordia University, 
		Montreal, Quebec, Canada, email:  sy\_al@encs.concordia.ca, jiayuan.yu@concordia.ca},
		Robert Shorten$^\dagger$\thanks{$^\dagger$School of Electrical, Electronic
		and Communications Engineering, University College Dublin, Dublin, Ireland, email: robert.shorten@ucd.ie},
	Fabian Wirth$^\ddagger$\thanks{$^\ddagger$Faculty of Computer Science and Mathematics, University
		of Passau, Passau, Germany, email:  fabian.wirth@uni-passau.de}, and
	Jia Yuan Yu$^\ast$
}

\begin{document}

\maketitle
\thispagestyle{empty}
\pagestyle{empty}

\begin{abstract}
	In several smart city applications, multiple resources must be
allocated among competing agents that are coupled through such shared resources and are constrained --- either through limitations of communication
infrastructure or privacy considerations. We propose a distributed algorithm to solve such
	distributed multi-resource allocation problems with no direct
	inter-agent communication. We do so by extending a recently introduced
	additive-increase multiplicative-decrease (AIMD) algorithm, which only
	uses very little communication between the system and agents. Namely,
	a control unit broadcasts a one-bit signal to agents whenever one of
	the allocated resources exceeds capacity. Agents then respond to this
	signal in a probabilistic manner. In the proposed algorithm, each
	agent makes decision of its resource demand locally and an agent is unaware of the resource allocation of other agents. In empirical results, we observe that the average allocations converge over time to optimal allocations.
\end{abstract}
\let\thefootnote\relax\footnotetext{To appear in IEEE International Smart Cities Conference (ISC2 2018), Kansas City, USA, September, 2018.  }	
\section{Introduction}
Smart cities are built on smart infrastructures like intelligent transportation systems, security systems, smart grids, smart hospitals, smart waste management systems, etc., \cite{Harrison2010, Zanella2014}. Internet of things (IoT) are the essential building blocks to develop such smart infrastructures \cite{Hernandez2011, Mohanty2016}, we call these devices as {\em Internet-connected devices (ICDs)}. In several smart city applications, multiple resources must be allocated among competing Internet-connected devices that are coupled through multiple resources. Generally speaking, such problems are more difficult to solve than those with a
single resource. This is particularly true when Internet-connected devices are constrained --- either through limitations of communication infrastructure, or due to privacy considerations. These distributed optimization problems have numerous applications in smart cities and other application areas.
The recent literature is rich with
algorithms that are designed for distributed control and optimization
applications. While this body of work is too numerous to enumerate, we point the interested readers to the works of Nedic \cite{Nedic2009},\cite{Nedic2011}; Cortes \cite{KIA2015}; Jadbabaie and
Morse\cite{Jadbabaie2003}; Bullo \cite{Bullo2011}; Pappas
\cite{Pappas2017}, Bersetkas~\cite{Bertsekas2011}; Tsitsiklis
\cite{Blondel2005} for recent contributions.  A survey of some of
the related work is given in \cite{Wirth2014}.  

 In many instances in smart cities and other areas, network  of Internet-connected devices
achieve optimal allocation of resources through regular communication
with each other and/or with a control unit. Motivated by such
scenarios, we propose an algorithm that is tailored for these but does not require inter-device communication due to privacy considerations.
 The proposed solution is based on the generalization of stochastic {\em additive-increase and multiplicative-decrease (AIMD)} algorithm \cite{Wirth2014}. By way of background, the AIMD algorithm was proposed in the context of congestion avoidance in transmission control protocol (TCP) \cite{Chiu1989}. The AIMD algorithm is further explored and used in several application domains for example, micro-grids \cite{Crisostomi2014}; multimedia \cite{Cai2005}; electric vehicle (EV) charging \cite{Studli2012}; resource allocation \cite{Avrachenkov2017}, etc. Interested readers can refer the recent book by Corless et al.~\cite{Corless2016} for an overview of some of the applications. The authors of \cite{Wirth2014} demonstrate that simple algorithms from Internet congestion control can be used to solve certain optimization problems. 
 Roughly speaking, in \cite{Wirth2014}, the iterative distributed optimization
algorithm works as follows. Internet-connected devices continuously acquire an increasing
share of the shared resource, this phase is called {\em additive increase} phase. When the aggregate resource demand of Internet-connected devices exceeds the total capacity of resource, then the control unit broadcasts a one bit
{\em capacity event} notification to all competing Internet-connected devices and these devices
respond in a probabilistic manner to reduce the demand, this phase is called {\em multiplicative decrease} phase. By judiciously selecting the probabilistic manner in which Internet-connected devices respond, a portfolio of
optimization problems can be solved in a stochastic and distributed
manner.

 Our contribution here is to demonstrate
that the ideas therein  \cite{Wirth2014} extend to a much broader (and more useful)
class of optimization problems which can be used in many application domains of smart cities and other areas. Our proposed algorithm builds
on the choice of probabilistic response strategies described therein
but is different in the sense that we generalize the approach to deal
with {\em multiple resource constraints} and the cost functions are coupled through multiple resources. We show that the optimal values obtained by proposed algorithm is same as if the optimization problem is solved in a centralized way. 

In the proposed solution, for a system with $m$ resources, in the worst case scenario the communication overhead is $m$ bits per time unit, which is very low. We would also like to mention that in the proposed solution, the communication complexity is independent of the number of Internet-connected devices competing for resources in the system. In this paper, we present a use case of a smart city that deploys a multi-camera coordination system, in which several cameras coordinate for the surveillance of the city. Each camera has private cost function which is coupled through allocation of multiple resources. Notice that we use the names agent and Internet-connected device interchangeably in this paper.
%

The paper is organized as follows, Section \ref{prob_form} describes the problem and provides the formulation of the problem, it also describes the conditions for optimality. A brief description of classical AIMD algorithm is presented in Section \ref{prelim}. Section \ref{divisible_mul_res} describes the multi-resource allocation strategies. The numerical results are presented in Section \ref{results}. The paper concludes with future directions in Section \ref{conc}. 

\section{Problem formulation}  \label{prob_form}
Suppose that a smart city deploys a {\em multi-camera coordination system} described in Figure \ref{Diag_camera}, in which several cameras work together for the surveillance of the city, these cameras are deployed at different locations. If a camera observes any unusual activity then it should demand the required amount of resources with higher probability than other cameras, to notify the observed activity immediately. Suppose that there are central servers set up by the city, which store and process the videos sent by all the cameras, these servers also act as a control unit. Each camera requires different amount of network bandwidth, CPU cycles, memory (RAM) and storage to transmit, process and store the videos on the central servers. Assume that a camera decides its demand based on its cost function and its previous allocations.
\begin{figure}[H]
	\centering
	\includegraphics[width=0.6\textwidth,clip=true,trim=8.5cm 11.5cm .025cm 4.2cm]{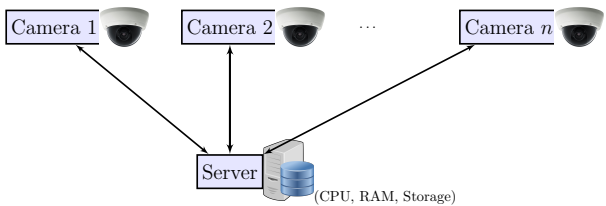}
	\caption{Multi-camera coordination system}
	\label{Diag_camera}
\end{figure}
Suppose that there are $n$ Internet-connected devices that compete for $m$ resources
$R^1, R^2, \ldots, R^m$ with capacity $C^1, C^2, \ldots, C^m$,
respectively. In this paper we assume that these Internet-connected devices are the cameras that compete for  memory (say $R^1$), storage ($R^2$) and network bandwidth ($R^3$). We further assume that each camera receives enough CPU cycles to process its data, for the sake of generality we use $m$ resources here. We denote $\mathcal{N}:=\{1, 2, \ldots, n\}$ and $\mathcal{M}:=\{1, 2, \ldots, m\}$ and use $i \in \mathcal{N}$ as an index for cameras and $j \in \mathcal{M}$ to index the resources. We assume that each camera has a private cost
function $f_i: \mathbb{R}^m \to \mathbb R$ which associates a cost to a
certain allotment of resources. We
assume that $f_i$ is twice continuously differentiable, convex, and
increasing in all variables, for all $i$. For all $i$ and $j$, we denote by $x_i^j \in \mathbb{R}_+$ the
amount of resource $R^j$ allocated to camera $i$. We are interested in the
following optimization problem of {\em multi-resource}
allocation:
\begin{align} \label{obj_fn1}
\begin{split}
\min_{{x}^1_1, \ldots, {x}^m_n} \quad &\sum_{i=1}^{n} f_i(x^1_i, x^2_i,
\ldots, x^m_i),    		\\
\mbox{subject to} \quad
&\sum_{i=1}^{n} x^j_i = C^j, \quad j \in \mathcal{M},		\\
&x^j_i \geq 0, \quad i \in \mathcal{N}, \ j \in \mathcal{M}.
\end{split}
\end{align}
Note that there are $nm$ decision variables $x^j_i$ in this optimization problem, for all $i$ and $j$.
We denote the solution to the minimization problem by
$x^{*} \in \mathbb{R}_+^{nm}$, where $x^* = (x_1^{*1}, \ldots, x_n^{*m})$. By compactness of the constraint set optimal solutions exist. We also assume strict convexity of the cost
function $\sum_{i=1}^{n} f_i$, so that the optimal solution is unique.

Suppose that $\mathbb{N}$ denotes the set of natural numbers and $k \in
\mathbb{N}$ denotes the time steps. To this end, we denote by $x_i^j(k)$ and $\overline{x}^j_i(k)$
 (refer \eqref{average_eqn}) the amount of resource allocated and average allocation at the (discrete) time step $k$, respectively. The camera can obtain any amount in $[0, C^j]$, for all $j$. We define the
average allocation for $i \in \mathcal{N}, \ j \in \mathcal{M}$, and $k \in \mathbb{N}$, as follows: 
\begin{align} \label{average_eqn}
\overline{x}^j_i(k)=\frac{1}{k+1} \sum_{\ell=0}^k x^j_i(\ell).
\end{align}
The goal is to propose a distributed iterative scheme, such that the long-term average allocations converge to the optimal allocations:
\begin{align}
\label{eq:longtermopt}
\lim_{k\to\infty} \overline{x}_i^j(k) \to x_i^{*j}, \quad \text{for $i \in \mathcal{N}$ and $j \in \mathcal{M}$}.
\end{align}

Let $\nabla_j f_i(.)$ be (partial) derivative of $f_i(.)$ with respect to resource $R^j$. 
 Similar to \cite{Syed2018}, we write the Lagrange multipliers of \eqref{obj_fn1}, with careful analysis we obtain that the derivatives of cost functions of all cameras competing for a particular resource should make a consensus at optimal allocations, i.e., the following holds true: 
\begin{align}\label{optimality}
\nabla_j f_i \big( x_i^{*1}, \ldots, x_i^{*m} \big) = \nabla_j f_u \big ( x_u^{*1}, \ldots, x_u^{*m} \big ), \nonumber\\ \mbox{ for all } i,u \in \mathcal{N} \mbox{ and } j \in \mathcal{M},
\end{align} 
 which satisfies all the Karush-Kuhn-Tucker (KKT) conditions. KKT conditions are necessary and sufficient condition for optimality of convex problem \eqref{obj_fn1}, interested readers may refer Chapter 5.5.3 \cite{Boyd2004} for a detailed discussion on KKT conditions. 
Now, to check the efficacy of our results we use the consensus of derivatives of cost function of all cameras with respect to a particular resource and show that the average allocation converge to the optimal allocation. 
\section{A primer on AIMD} \label{prelim}
The AIMD algorithm is of interest because it can be tuned to achieve
optimal distribution of a single resource among a group of agents. To this
end no inter-agent communication is necessary. The agents just receive
capacity signals from a control unit and respond to it in a stochastic
manner. This response can be tuned so that the long-term average
optimality criterion (cf. \eqref{eq:longtermopt}) can be achieved. The following is an excerpt from \cite{Corless2016}.

In AIMD algorithm each agent follows two rules of action at each time step:
either it increases its share of the resource by adding a fixed amount while total demand is
less than the available capacity, or it reduces its share in a
multiplicative manner when notified
that global capacity has been reached. In the additive increase (AI) phase of the algorithm agents probe the
available capacity by continually increasing their share of the
resource. The multiplicative decrease\index{multiplicative-decrease}
(MD) phase occurs when agents are notified that the capacity limit has
been reached; they respond by reducing their shares, thereby freeing
up the resource for further distribution. This pattern is repeated by every agent as
long as the agent is competing for the resource. The only information
given to the agents about availability of the resource is a
notification when the collective utilization of the resource achieves
some capacity constraint.  At such times, so called {\em capacity
	events}, some or all agents are instantaneously informed that
capacity has been reached.  The mathematical description of the basic continuous-time AIMD model
is as follows. 
Let $n$ agents compete for a resource, and suppose that $x_i(t) \in \mathbb{R}_+$ denotes the quantity of the collective resource obtained by agent $i$ at time $t\in \mathbb{R}_+$. Let $C$ denotes the total capacity of the resource available to the
entire system (which need not be known by the agents). The capacity constraint requires that
$\sum_{i=1}^{n} x_i(t) \le C$ for all $t$. As all agents are continuously
increasing their share this capacity constraints will be reached
eventually. We denote the times at which this happens by $t_k, k\in \mathbb{N}$.
At time $t_k$ the global utilization of the resource reaches
capacity, thus
$\sum_{i=1}^n x_i(t_k) = C.$
When capacity is achieved, some agents
decrease their share of the resource.
The instantaneous decrease of the share for agent $i$ is
defined by:
\begin{align} \label{MD} x_i(t_k^{+}) := \lim_{t \rightarrow t_k,
	\, t > t_k} x_i(t) = \beta_i x_i(t_k),
\end{align}
where $\beta_i$ is a constant satisfying $ 0\le \beta_i <1.$
In the simplest version of the algorithm, agents are assumed to
increase their shares at a constant rate in the AI phase: 
\begin{align} \label{AI} x_i(t ) = \beta_i x_i(t_k) + \alpha_i (t
- t_k), \quad t_k < t \le t_{k+1},
\end{align} 
where, $\alpha_i>0$,
is a positive constant, which may be different for different
agents, $\alpha_i$ is known as the {\em growth rate}\index{growth
	rate} for agent $i$. By writing $x_i(k)$ for the $i$th agent's
share at the $k$th capacity event as $x_i(k) := x_i(t_k)$ we have:
\begin{align*}
x_i(k+1) = \beta_i x_i(k) + \alpha_i T(k),
\end{align*}
where $
T(k) := t_{k+1}-t_k,$
is the time between events $k$ and $k+1$.  There are
situations where not all agents may respond to every capacity
event. Indeed, this is precisely the case considered in this paper. In
this case agents respond asynchronously to a congestion notification
and the AIMD model is easily extended by using our previous formalism
by changing the multiplicative factor\index{multiplicative factor} to
$\beta_i = 1$ at the capacity event if agent $i$ does not
decrease. 

\section{Multi-resource allocation} \label{divisible_mul_res}

Let $\delta_j>0$ be a fixed constant, for all $j$ and $\nabla^2 f$ be the matrix of second order partial derivatives of $f$ called Hessian of $f$. Furthermore, let $\mathcal{F}_\delta$ denotes the set of twice continuously differentiable functions defined as follows: 
\begin{equation} \label{def_F_delta}
\begin{split}
\mathcal{F}_\delta = \Big\{ f:\mathbb R^m_+ \to \mathbb R
\Big \lvert   \Big( x^j > 0 \implies   0 < \delta_j \nabla_j f(x) < x^j \\
\text{ for all $j$}\Big)  \text{ and } \nabla^2 f(x) \succeq 0 \text{ for all }x\in \mathbb{R}^m_+ \Big\}.
\end{split}
\end{equation}
Here, $\nabla^2 f(x) \succeq 0$ represents a positive semi-definite matrix. We observe that $\mathcal{F}_\delta$ is essentially the set of functions that are convex, twice continuously differentiable and increasing in each coordinate.
We consider the problem of allocating $m$ resources with capacity
$C^j$, for $j\in \mathcal{M}$ among $n$ competing Internet-connected devices, whose cost functions $f_1,\ldots,f_n$ belong to the
set $\mathcal F_\delta$. Additionally, each cost function is private and should be
kept private. However, we assume that the set $\mathcal F_\delta$ is
common knowledge --- the control unit needs the knowledge of $\delta_j$ and the
Internet-connected devices need to have cost functions from this set. We should make clear that $\mathcal F_\delta$ has a large range of allowed cost functions. By knowing this range, the control unit can not easily guess the actual cost function, thereby giving the Internet-connected device a nontrivial amount of privacy.
 In this paper, we propose a distributed algorithm that determines instantaneous allocations $\{x_i^j(k)\}$, for all $i, j$ and $k$. Recall that $x^* = (x_1^{*1}, \ldots, x_n^{*m})$ is the solution of \eqref{obj_fn1}. We also show empirically that
for every Internet-connected device $i$ and resource $R^j$, the long-term average allocations converge to the optimal allocations i.e.,
$\overline{x}_i^j(k) \to {x}_i^{*j}$
as $k\to \infty$ (cf. \eqref{eq:longtermopt}) to achieve the minimum overall cost to the society called {\em social cost}.

\subsection{Algorithm}
In the system, each Internet-connected device runs a distinct distributed AIMD algorithm. We use $\alpha^j>0$
to represent the additive increase factor or growth rate and $0 \leq \beta^j \leq 1$ to
represent multiplicative decrease factor, both corresponding to
resource $R^j$, for $j \in \mathcal{M}$.	
We represent $\Gamma^j$ as the {\em normalization factor}, chosen based on the knowledge of fixed constant $\delta_j$ to scale probabilities $\lambda_i^j(k)$.
Every algorithm is initialized with the same set of parameters 
$\Gamma^j$, $\alpha^j$, $\beta^j$ received from the control unit of the system.   
We represent the one-bit \emph{capacity event signals} by $S^j(k) \in \{0,
1\}$ at time step $k$ for resource $R^j$,
for all $j$ and $k$.
At the start of the system the control unit initializes the capacity event signals $S^j(0)$ with $0$, and updates $S^j(k)=1$ when the total allocation $\sum_{i=1}^n x_i^j(k)$ exceeds the capacity $C^j$ of a resource $R^j$ at a time step $k$. After each update, control unit broadcasts it to Internet-connected devices in the system signaling that the total demand has exceeded the capacity of the resource $R^j$. We describe the algorithm of control unit in Algorithm \ref{algoCU1}.
\begin{algorithm}  \SetAlgoLined Input:
	$C^{j}$, for $j \in \mathcal{M}$.
	
	Output:
	$S^{j}(k+1)$, for $j \in \mathcal{M}$, $k \in
	\mathbb{N}$.
	
	Initialization: $S^{j}(0) \leftarrow 0$, for $j \in \mathcal{M}$,
	
	broadcast $\Gamma^{j} \in (0,\delta_j]$ according to \eqref{gamma};
	
	\ForEach{$k \in \mathbb{N}$}{
		
		\ForEach{$j \in \mathcal{M}$}{
			\uIf{ $\sum_{i=1}^{n} {x}_i^{j}(k) > C^{j}$}{
				$S^{j}(k+1) \leftarrow 1
				$\;
				
				broadcast $S^{j}(k+1)$;	
			}
			
			\Else{$S^{j}(k+1) \leftarrow 0$\;
	} }}
	
	\caption{Algorithm of control unit}
	\label{algoCU1}
\end{algorithm}

The algorithm of each Internet-connected device works as follows.  At every time
step, each algorithm updates its demand for resource $R^j$ in one of the
following ways: an {\em additive increase (AI)} or a {\em multiplicative
	decrease (MD) phase}.
In the additive increase phase, the algorithm increases its demand
for resource $R^j$ linearly by the constant $\alpha^j$ until it
receives a capacity event signal $S^j(k) =1$ from the control unit of
the system at time step $k$ that is:
\begin{align*}
x_i^j (k+1) = 	x_i^j (k) + \alpha^j.
\end{align*}  
The multiplicative decrease phase occurs when total demand exceeds the capacity of a resource (say $R^j$), and the control unit in response broadcasts a capacity event signal $S^j(k) =1$.
In turn, each Internet-connected device $i$ responds with probability $\lambda^j_{i}(k)$ by scaling its demand by $\beta^j$. If $S^j(k) =1$, we thus have:
\begin{align*}  
x_i^j(k+1)= \left\{
\begin{array}{ll}
\beta^j x_i^j(k) & \mbox{with probability } \lambda^j_{i}(k) , \\
x_i^j(k) & \mbox{with probability } 1-\lambda^j_{i}(k).\\
\end{array}
\right.
\end{align*}	
The probability $\lambda^j_{i}(k)$
depends on the average allocation and the derivative of cost function with respect to $R^j$ of Internet-connected device $i$,
for all $i$ and $j$.  It is calculated as
follows:
\begin{align} \label{prob_x} \lambda^j_{i}(k) = \Gamma^j  
\frac{{\nabla_j} f_i \big( \overline{x}^1_i(k), \overline{x}^2_i(k), \ldots,
	\overline{x}^m_i(k) \big)}{\overline{x}^j_i(k)},
\end{align}
for all $i $, $j$ and $k$. After the reduction of demands, all Internet-connected devices can
again start to increase their demands until the next capacity event occurs. This process repeats. 
It is obviously required that always $ 0 < \lambda^j_{i}(k) < 1$. To
this end the normalization factor $\Gamma^j$ is
needed which is based on the set $\mathcal F_\delta$. The fixed constant $\delta_j >0$ is chosen such that $\Gamma^j$ satisfies the following:
\begin{align}\label{gamma_delta}
0 <\Gamma^j \leq \delta_j, \text{ for all } j.  
\end{align}
At the beginning of
the algorithm the normalization factor $\Gamma^j$ for
resource $R^j$ is calculated explicitly as the following and broadcast to all Internet-connected devices in the system:
\begin{align}\label{gamma}
\Gamma^j = 
\inf_{x_1^1,\ldots,x_n^m \in \mathbb{R}_+,  f \in
	\mathcal{F}_\delta}
\Big(\frac{x^j}{\nabla_j f(x^1,
	x^2, \ldots, x^m)} \Big), \text{ for all } j.  
\end{align}
To capture the stochastic nature of the response to the capacity signal, we define the following independent Bernoulli random variables:
\begin{align} \label{bern_var}
b^j_i(k)=
\left\{
\begin{array}{ll}
1  & \mbox{with probability } \ \lambda^j_{i}(k),\\
0 & \mbox{with probability } 1-\lambda^j_{i}(k), 
\end{array}
\right. 
\end{align}
for all $i$, $j$ and $k$. The following theorem proves that $ 0 < \lambda^j_{i}(k) < 1$.   
\begin{theorem}[Probability measure] \label{theorem1}  For a given $\delta_j >0$, if \text{ } $\overline{x}_i^j(k) >0$ and the cost function $f_i$ of Internet-connected device $i$ belongs to $\mathcal{F}_\delta$, then for all $i, j$ and $k$, $\lambda_i^j(k)$ satisfies $0 < \lambda_i^j(k) < 1$.
\end{theorem}
\begin{proof}
	It is given that $f_i \in \mathcal{F}_\delta$ and $\overline{x}_i^j(k) > 0$ for all $i$, $j$ and $k$ then from \eqref{def_F_delta}, we write as follows:
	\begin{align} \label{eq117}
	0 < \delta_j \nabla_j f_i \big( \overline{x}_i^1(k),\overline{x}_i^2(k), \ldots, \overline{x}_i^m(k) \big) < \overline{x}_i^j(k).
	\end{align}
	We know that for a fixed constant $\delta_j>0$, the normalization factor $\Gamma^j$ satisfies $0 < \Gamma^j \leq \delta_j$,  for all $j$  (cf. \eqref{gamma_delta}). It is given that $\overline{x}_i^j(k) >0$, dividing \eqref{eq117} by $\overline{x}_i^j(k)$ and substituting $\Gamma^j$ we obtain as follows:
	\begin{align} \label{eq114}
	&0 <   \frac{\Gamma^j \nabla_j f_i \big( \overline{x}_i^1(k),\overline{x}_i^2(k), \ldots, \overline{x}_i^m(k) \big)}{\overline{x}_i^j(k)} < 1, \\&\text{ for all $i,j$ and $k$} \nonumber .
	\end{align}
	Since, for all $i, j$ and $k$, an Internet-connected device $i$ makes a decision to respond the capacity event of a resource $R^j$ with $\lambda_i^j(k)$ (cf.  \eqref{prob_x}).
	Hence, after placing $\lambda_i^j(k)$ in \eqref{eq114}, we obtain
	$0 <  \lambda_i^j(k) < 1$, for all $i,j$ and $k$. 
\end{proof}	
Notice that because of the stochastic nature of the algorithm, an Internet-connected device may reduce its resource demand and fails to complete its current job, but only in cases where other Internet-connected devices derive more benefit than this Internet-connected device. This is done in order to maximize the overall benefit to the society called {\em social welfare}.
\begin{figure}[H]
	\centering
	\includegraphics[width=0.55\textwidth,clip=true,trim=7.5cm 9.05cm 1.25cm 4.2cm]{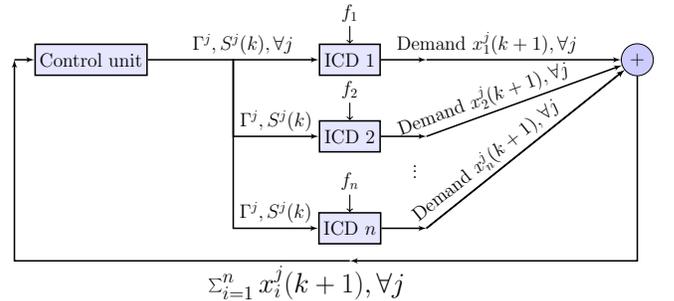}
	\caption{Block diagram of the proposed AIMD model, here the central server works also as the control unit, and ICD represents the algorithm of Internet-connected device.}
	\label{Diag_AIMD}
\end{figure}
We present the block diagram of the system in Figure~\ref{Diag_AIMD} and the proposed distributed multi-resource allocation algorithm for each Internet-connected device in Algorithm \ref{algo1}.
\begin{algorithm} \SetAlgoLined Input:
	$S^{j}(k)$, for $j \in \mathcal{M}, k \in
	\mathbb{N}$
	and $\Gamma^j$,
	$\alpha^j, \beta^j$, for $j \in \mathcal{M}$.
	
	Output:
	$x^j_i(k+1)$, for $j \in \mathcal{M}$, $k \in
	\mathbb{N}$.
	
	Initialization: $x^j_i(0) \leftarrow 0$ and
	$\overline{x}^j_i(0) \leftarrow x^j_i(0)$, for
	$j \in \mathcal{M}$;

	\While{Internet-connected device $i$ is active at $k \in \mathbb{N}$}{
		
		\ForEach{$j \in \mathcal{M}$}{

			\uIf{$S^{j}(k) = 1$}{
				$ \lambda^j_{i}(k) \leftarrow \Gamma^j
				\frac{{\nabla_j} f_i \left ( \overline{x}^1_i(k),
					\overline{x}^2_i(k), \ldots, \overline{x}^m_i(k)
					\right)}{\overline{x}^j_i(k)}$;
				
				generate independent Bernoulli random variable
				$b^j_i(k)$ with the parameter
				$\lambda^j_{i}(k)$;
				
				\uIf{ $b^j_i(k)=1$}{
					$x^j_i(k+1) \leftarrow \beta^j x^j_i(k)
					$;}
				
				\Else{ $x^j_i(k+1) \leftarrow x^j_i(k) $; }
				
			} \Else{
				$x^j_i(k+1) \leftarrow x^j_i(k) +
				\alpha^j$; }
			
			$\overline{x}^j_i(k+1) \leftarrow \frac{k+1}{k+2}
			\overline{x}^j_i(k) + \frac{1}{k+2} x^j_i(k+1);$
	} }
	\caption{Algorithm of Internet-connected device $i$ (AIMD $i$) }
	\label{algo1}
\end{algorithm}
We observe using numerical results in Section \ref{results} that the average
allocation $\overline{x}_i^j(k)$ converge to the optimal
allocation $x_i^{*j}$ of resource $R^j$ over time, for all $i$ and $j$.
\begin{remark}[Communication overhead] Suppose that there are $m$ resources in
	the system, then communication overhead will be
	$\sum_{j=1}^{m} S^{j}(k)$ bits at $k^{th}$ time step, for all
	$k$. In the worst case scenario this will be $m$ bits per
	time unit, which is quite low. Furthermore, the communication complexity
	does not depend on the number of Internet-connected devices in the system.
\end{remark}
%
\section{Numerical results} \label{results}
 In this section, we use the multi-camera coordination system described in Section \ref{prob_form}. We illustrate here that the proposed distributed multi-resource allocation algorithm provides optimal allocations to all cameras in long-term average allocations and the city achieves a minimum social cost, these optimal values are same as if the problem is solved in a centralized way. 

Now, suppose that there are $60$ cameras in the multi-camera coordination system, each camera has different resolution, frame size and frame generation rate (frames per second), therefore every camera generates different amount of data. For example, a camera with frame size of $30$ KB and frame rate $10$ frames per second, produces $300$ KB video data in one second, hence $1.08$ GB in an hour.  
%
%
 Let us assume that the videos from all the cameras are stored on a server or Cloud. To transmit, process and store the videos on the server or Cloud they require network bandwidth, CPU cycles, memory (RAM), and disk storage. We assume that each camera gets enough CPU cycles to process the data but the server has limited memory (RAM), disk storage and network bandwidth. Let, $R^1$ denotes the memory (RAM), $R^2$ denotes the disk storage and $R^3$ denotes the network bandwidth. We chose  capacities of memory, disk storage and network bandwidth as $C^1 = 32$ GB, $C^2 = 200$ GB and $C^3 = 250$ Mbps, respectively. Let $10$ GB is denoted by $\mathrm{GB}^D$ and $10$ Mbps is denoted by $\mathrm{Mbps}^D$, then we write $C^2 = 20$ $\mathrm{GB}^D$ and $C^3 = 25$ $\mathrm{Mbps}^D$, we do so for the sake of uniformity of cost of resources in the cost function. Let $f_i$ be the cost function of camera $i$, each cost function depends on the average allocation of the resources. Our aim is to minimize the total cost incurred in transmitting, storing and processing the video data. For illustrative purpose we use the pricing model of Google compute engine for custom machine types \cite{Google2018} as shown in Table \ref{price_tab}. We create a dynamic pricing scheme for our simulation, keeping the values of Table \ref{price_tab} into consideration. Notice that in Table \ref{price_tab}, for the disk storage we use the price of image storage for $10$ days and the listed prices are for Iowa state. Furthermore, we use the price of bandwidth for North America as listed in \cite{Prince2018}. 
%
%
\begin{table}[ht]
\caption{Pricing scheme of Google compute engine custom machines for $4$ hours}
\centering 
\begin{tabular}{c c c c}
\hline
\hline                        
Resource type &  Price per unit (USD)  \\[0.5ex] \hline
vCPU & 0.132696  \\
RAM (GB) & 0.017784  \\
Disk storage ($10$ GB) & 0.283333   \\
Network bandwidth (10 Mbps) \cite{Prince2018} & 0.277775    \\[1ex]
\hline
\end{tabular}
\label{price_tab}
\end{table}

 Now, let $a_i, b_i,$ and $c_i$ represent the price for RAM, disk storage and bandwidth, and $d_i$ represents any other costs incurred. For all $i$, let $a_i, b_i,$ $c_i$ and $d_i$ are modeled as uniformly distributed random variables. In the simulation, we use $a_i \in \{10, 11, \ldots, 20\}$, $b_i \in \{25, 26, \ldots, 35\}$, $c_i \in \{22, 23, \ldots, 32\}$ and $d_i \in \{1, 2, \ldots, 5\}$. We use these random variables to generate random costs of each camera at different time steps, as described in \eqref{cost_fn}. To take vCPU price into consideration, we add a fraction of its price in the price of memory.
In the simulation, we chose the following additive increase factors $\alpha^1=25$ MB, $\alpha^2 = 20$ MB and $\alpha^3 = 225$ Kbps. Additionally, we chose the following multiplicative decrease factors $\beta^1=0.70$, $\beta^2 = 0.85$ and $\beta^3 = 0.75$, for the respective resources. Furthermore, we use the normalization factors $\Gamma^1 = \Gamma^2 = \Gamma^3 = 1/90$. Notice that allocation $x_i^1$ is in GB, $x_i^2$ is in $\mathrm{GB}^D$ and $x_i^3$ is in $\mathrm{Mbps}^D$.

	\begin{strip}
	\begin{align} \label{cost_fn}
	f_{i}(x_i^1, x_i^2, x_i^3) =
	\begin{cases}
	 a_i(x_i^1)^2  + c_i(x_i^3)^2 + \frac{1}{2}a_i(x_i^1)^4 + 2b_i(x_i^2)^4 + \frac{1}{2}b_i(x_i^2)^6 + \frac{1}{4} c_i(x_i^3)^4 + \frac{1}{8}d_i(x_i^3)^8  & \text{ w.p. } 1/3 \\
	  a_i(x_i^1)^2 + b_i(x_i^2)^2 + \frac{1}{2}b_i(x_i^2)^4 + \frac{3}{2} c_i(x_i^3)^4 & \text{ w.p. } 1/3 \\
	 b_i(x_i^2)^2 + c_i(x_i^3)^2 +  \frac{1}{3} a_i(x_i^1)^6  + \frac{1}{6}d_i(x_i^2)^6 + \frac{1}{8}d_i(x_i^3)^4 & \text{ w.p. } 1/3. 
	\end{cases}
	\end{align}
\end{strip}
 			
 Here, for the illustrative purpose we use only few cost functions but the proposed algorithm works on a set of cost functions with condition that these are convex, twice differentiable and increasing functions.	
\begin{figure*}[h] 
	\centering
	\subfloat[]{%
		\includegraphics[width=0.33\linewidth]{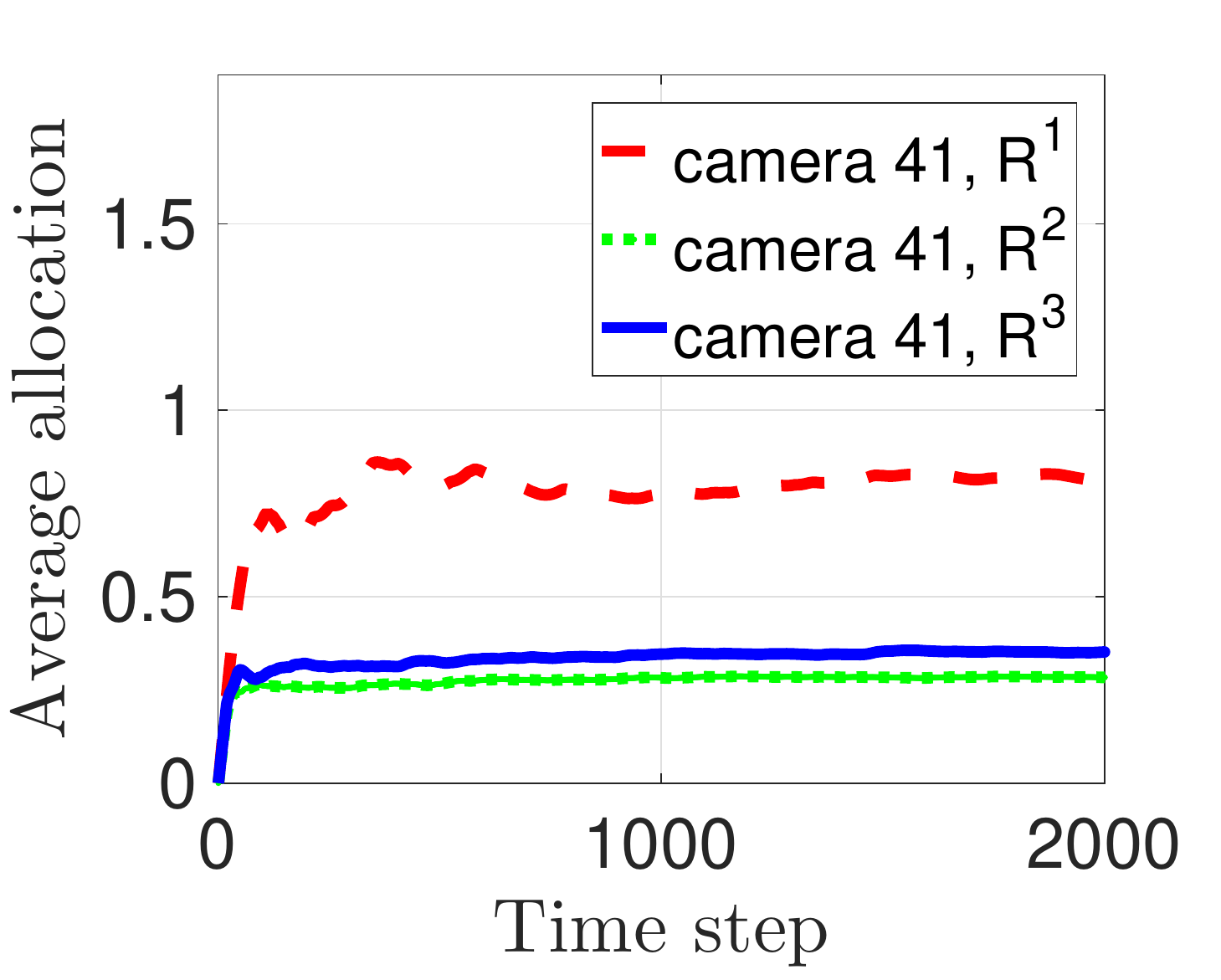}}
	\label{avg_aimd}\hfill
	\subfloat[]{%
		\includegraphics[width=0.33\linewidth]{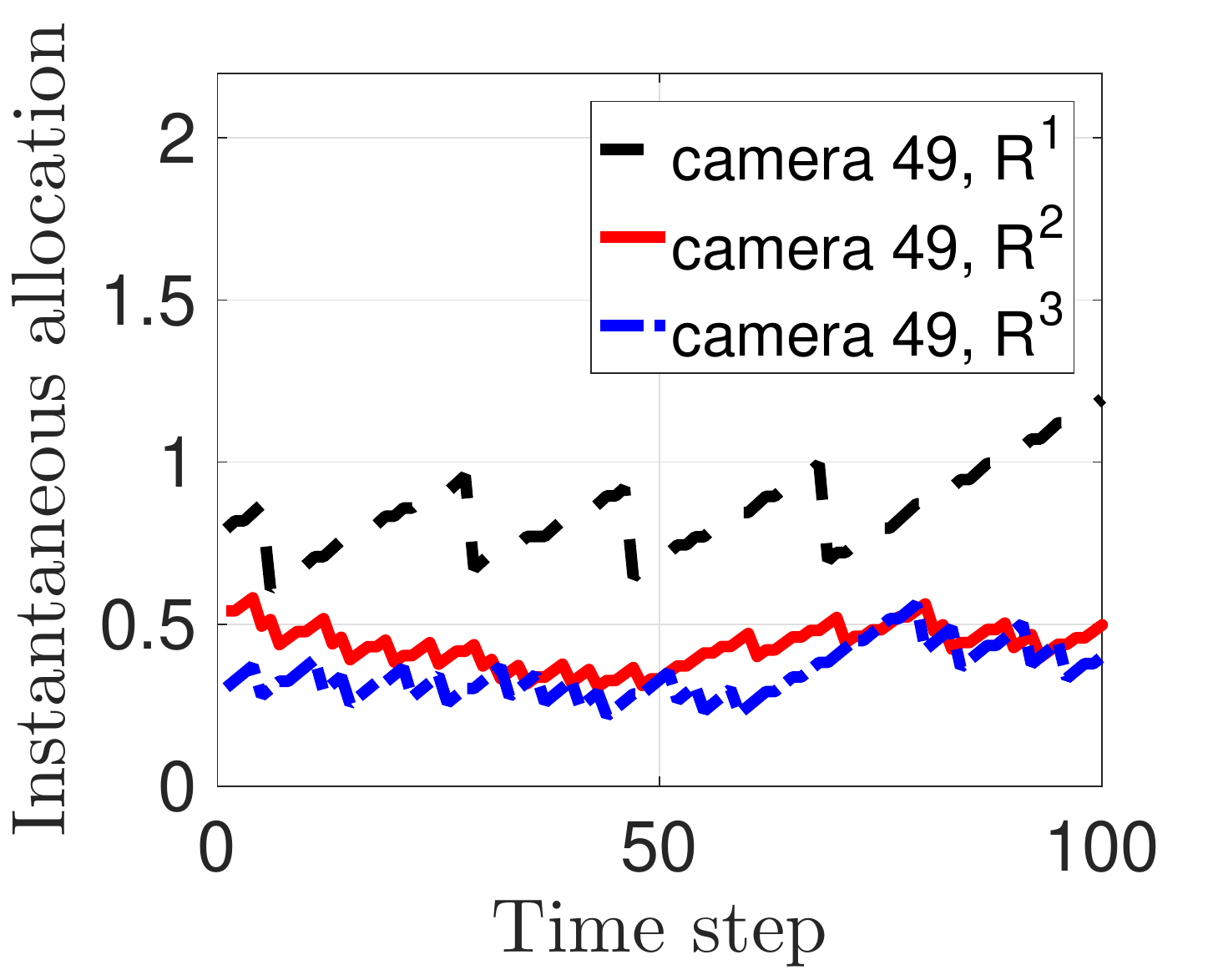}}
	\label{inst_alloc_aimd}\hfill
	\subfloat[]{%
		\includegraphics[width=0.33\linewidth]{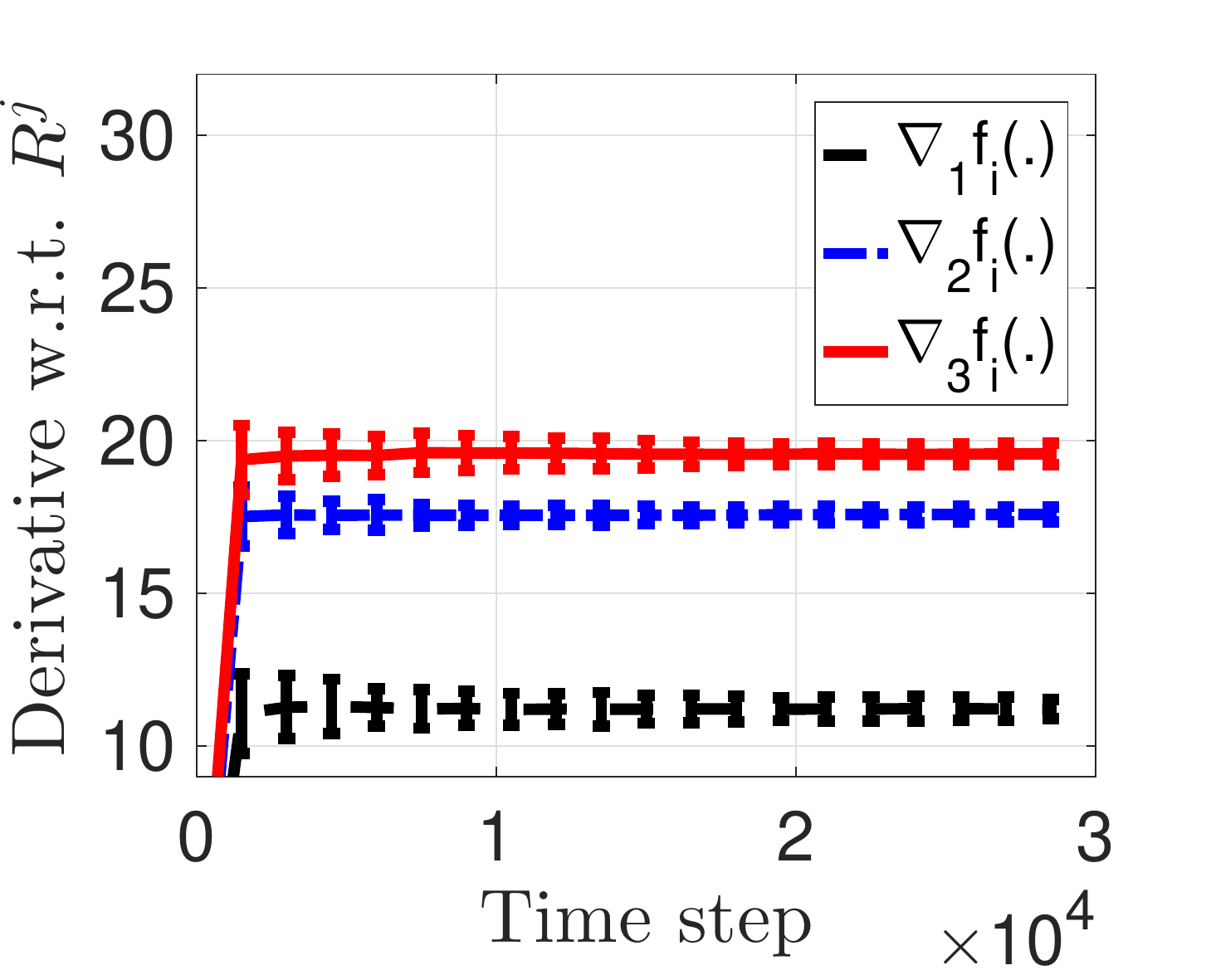}}
	\label{err_grad_aimd}\hfill
	
	\caption{(a) Evolution of average allocation of resources, (b) instantaneous allocation of resources for last $100$ time steps, (c) evolution of profile of derivatives of $f_i$ of all cameras}
	\label{fig1} 
\end{figure*}
	\begin{figure*}[h] 
	\centering
	\subfloat[]{%
		\includegraphics[width=0.345\linewidth]{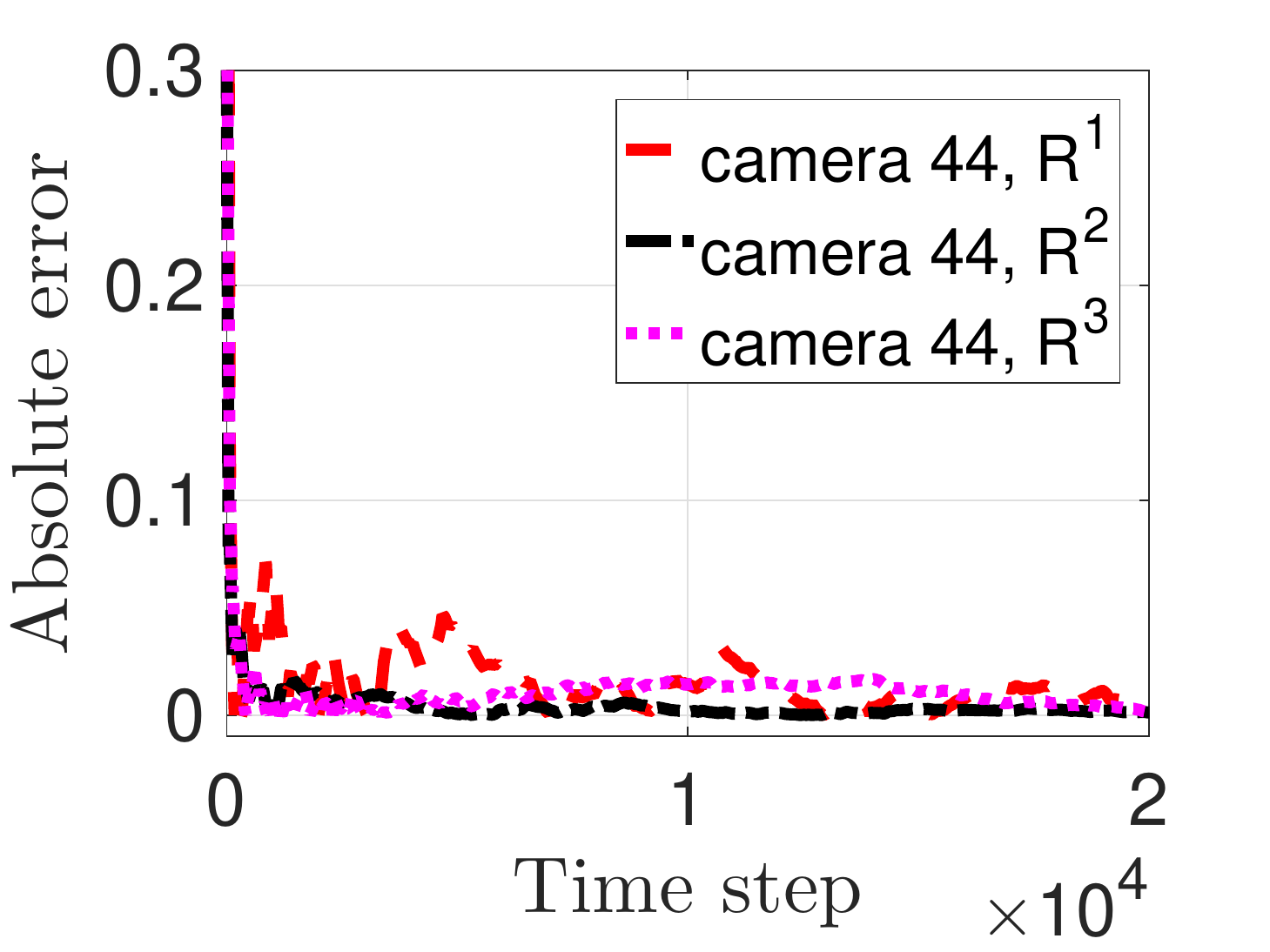}}
	\label{err_grad_BAIMD_2var_x1}\hfill
	\subfloat[]{%
		\includegraphics[width=0.345\linewidth]{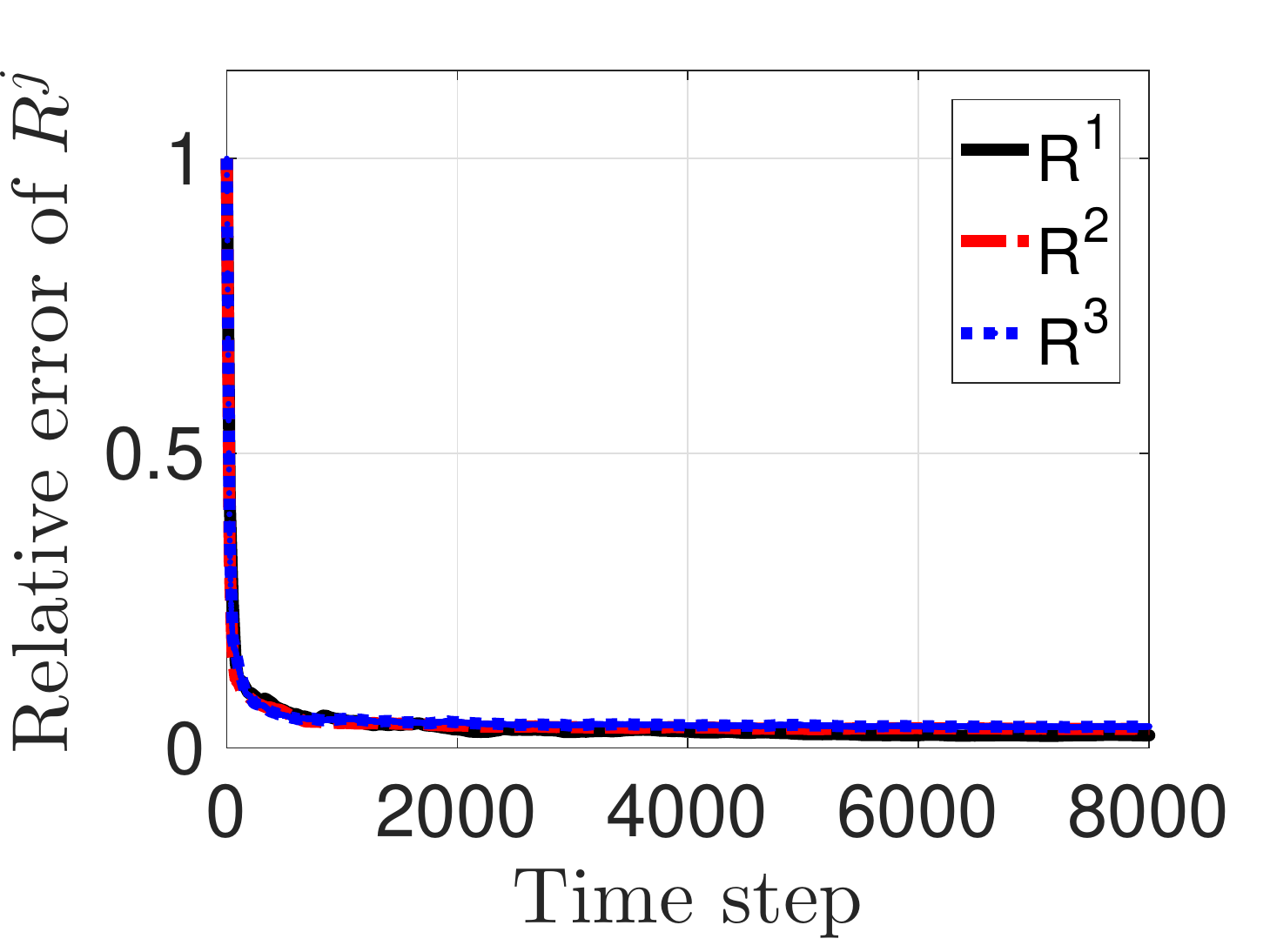}}
	\label{err_grad_BAIMD_2var_x2}\hfill
	\subfloat[]{%
		\includegraphics[width=0.308\linewidth]{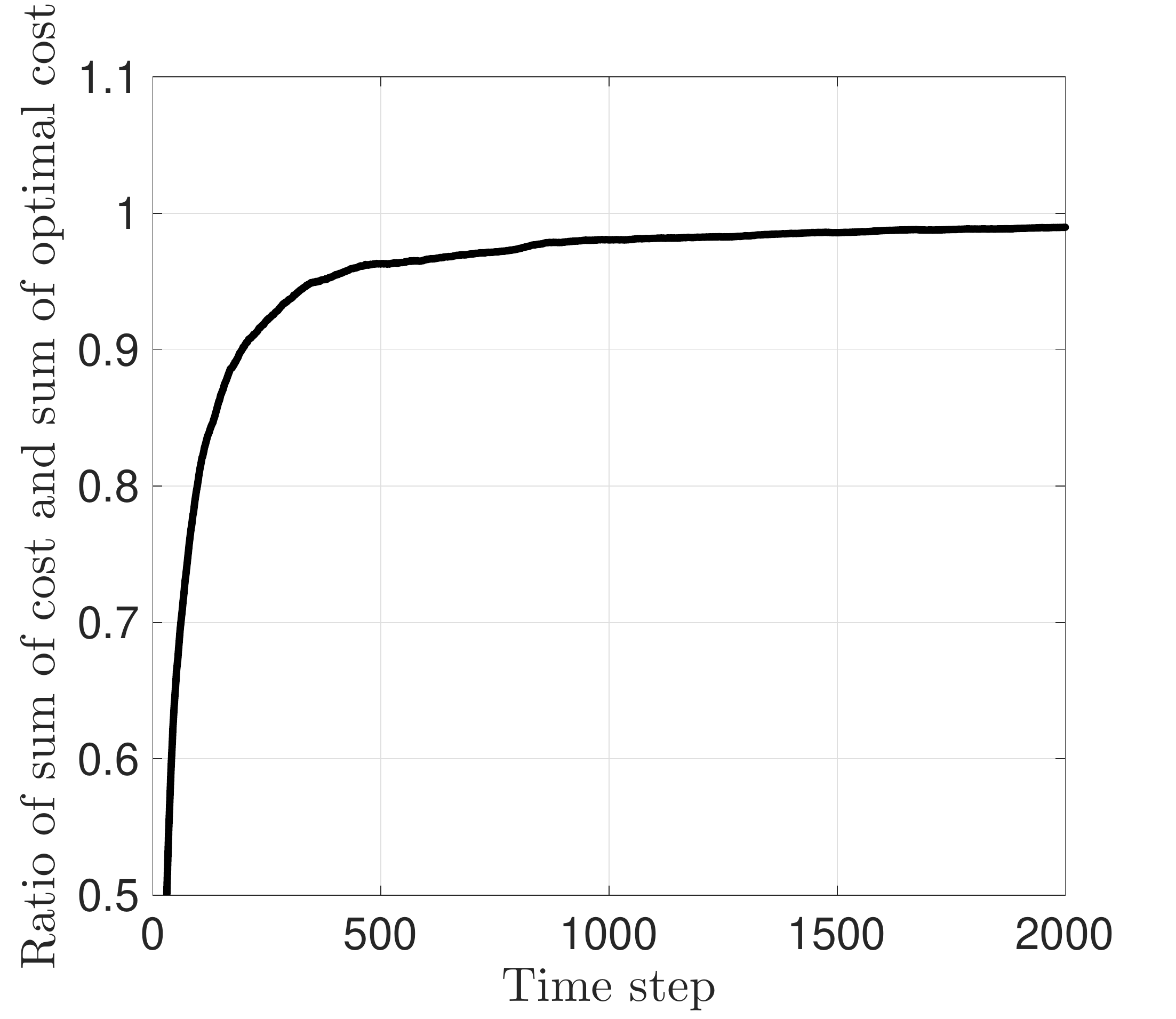}}
	\label{avg_BAIMD_2var}\hfill
	
	\caption{(a) Evolution of absolute difference between average allocation and the optimal allocation (calculated), (b) evolution of relative error of average allocation and the optimal allocation, (c) evolution of ratio of total cost and total optimal cost of all cameras}
	\label{fig2} 
\end{figure*}

The following are some of the
results obtained from the simulation. We observe in Figure \ref{fig1}(a) that the average allocations $\overline{x}_i^j(k)$ converge over time to its respective optimal value $x^{*j}_i$, for all $i$ and $j$. 
Figure \ref{fig1}(b) shows the instantaneous allocation $x_i^j(k)$ of all resources over last $100$ time steps,  which demonstrates the allocation phases (AI and MD). 
\begin{figure*}[h] 
	\centering
	\subfloat[]{%
		\includegraphics[width=0.33\linewidth]{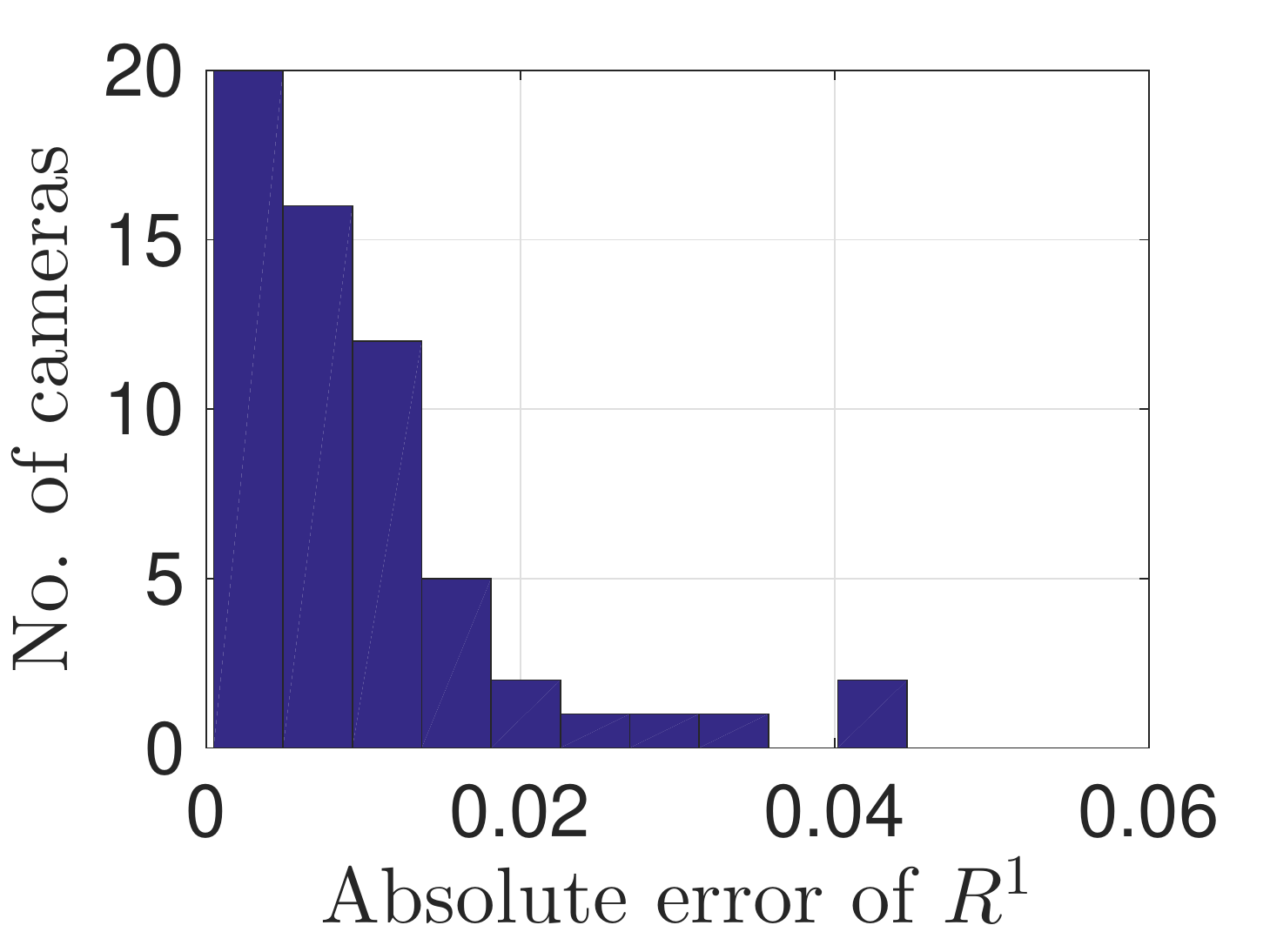}}
	\label{1a}\hfill
	\subfloat[]{%
		\includegraphics[width=0.33\linewidth]{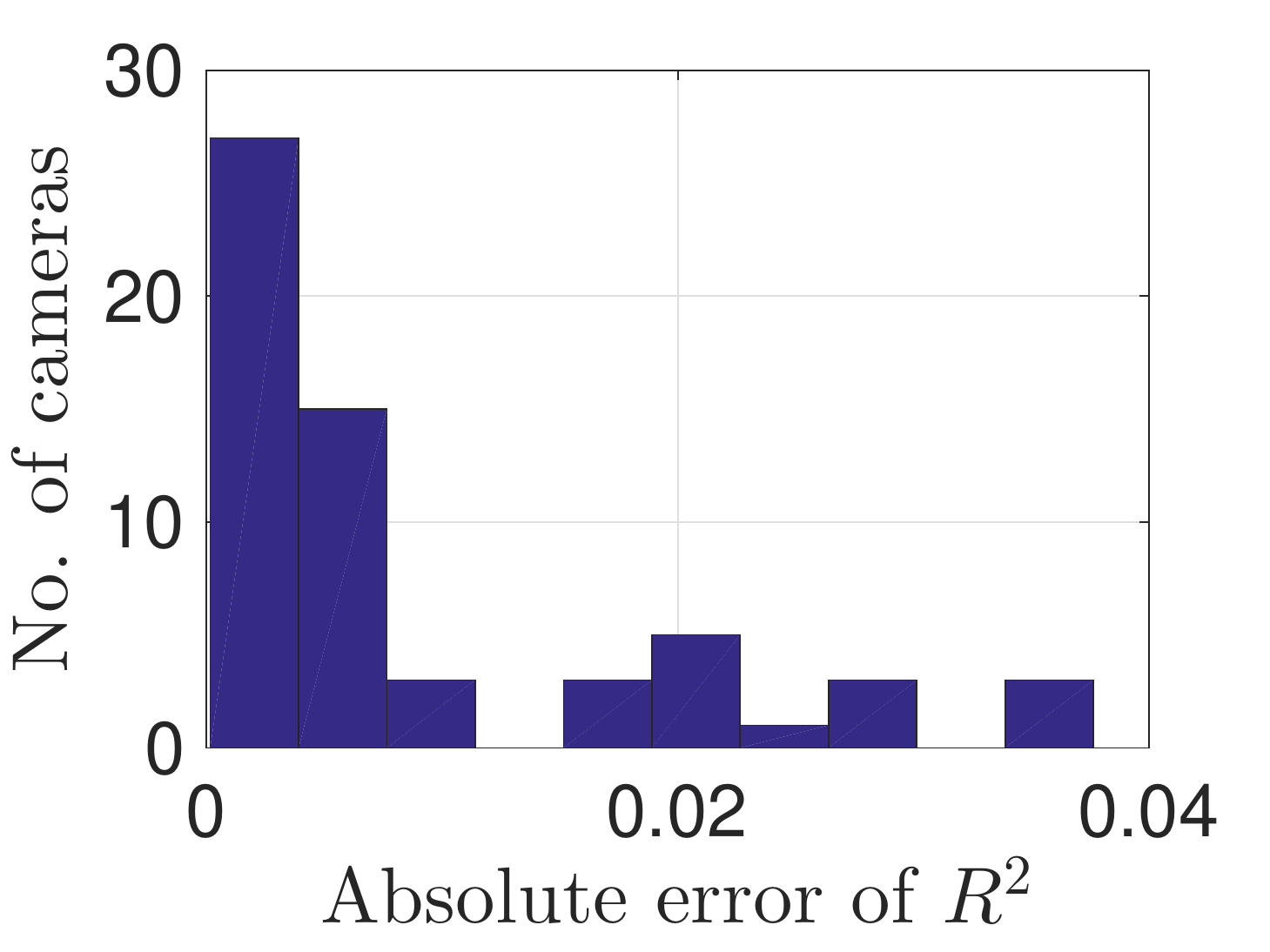}}
	\label{1b}\hfill
	\subfloat[]{%
		\includegraphics[width=0.33\linewidth]{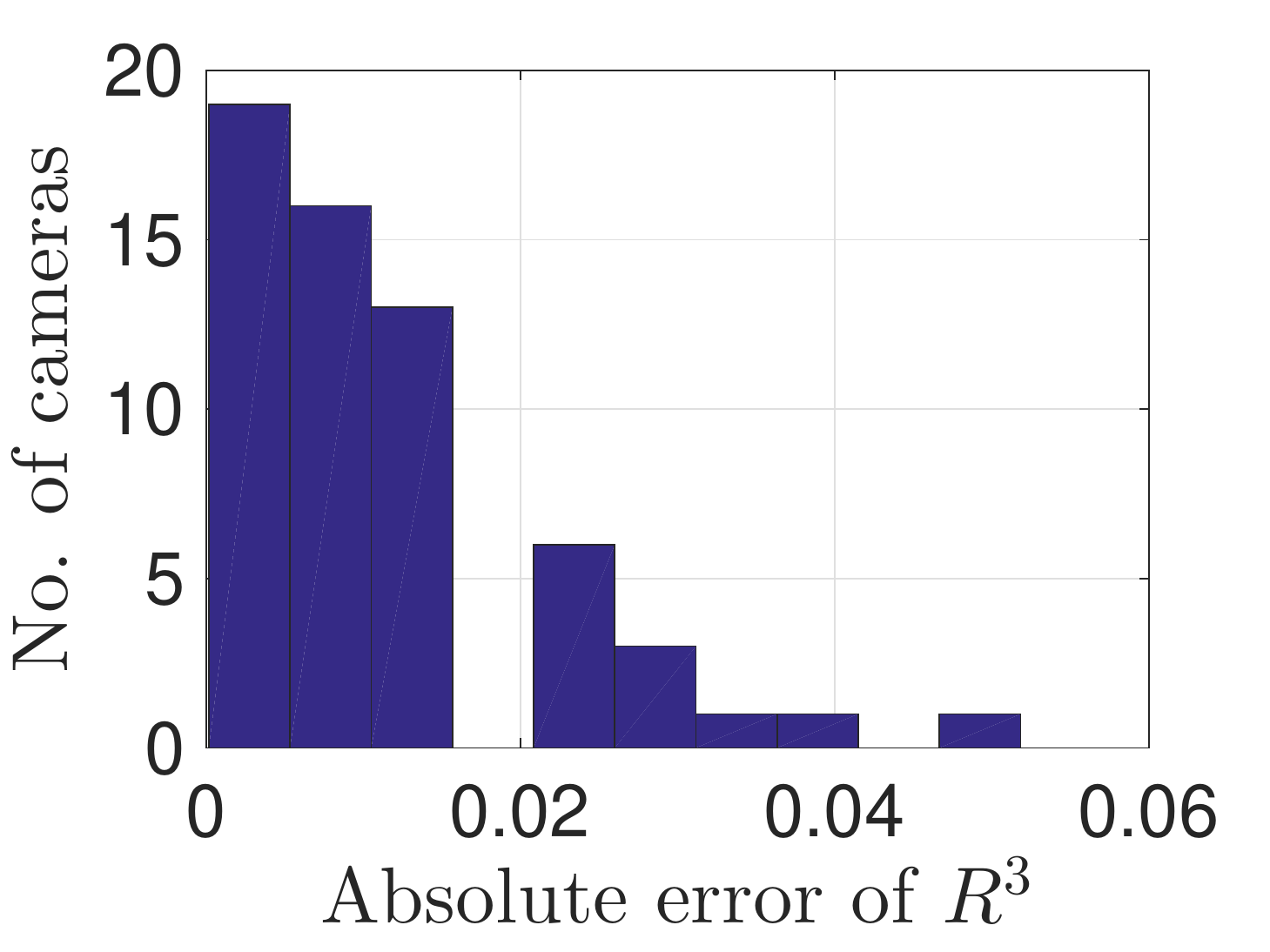}}
	\label{1c}\hfill
	\caption{Histogram of absolute error $|\overline{x}_i^j(K)-x_i^{*j}|$ of all cameras at time step $K=30000$, (a) absolute error of $R^1$, (b) absolute error of $R^2$, (c) absolute error of $R^3$}
	\label{fig3} 
\end{figure*}
\begin{figure*}[h] 
	\centering
	\subfloat[]{\includegraphics[width=0.33\linewidth]{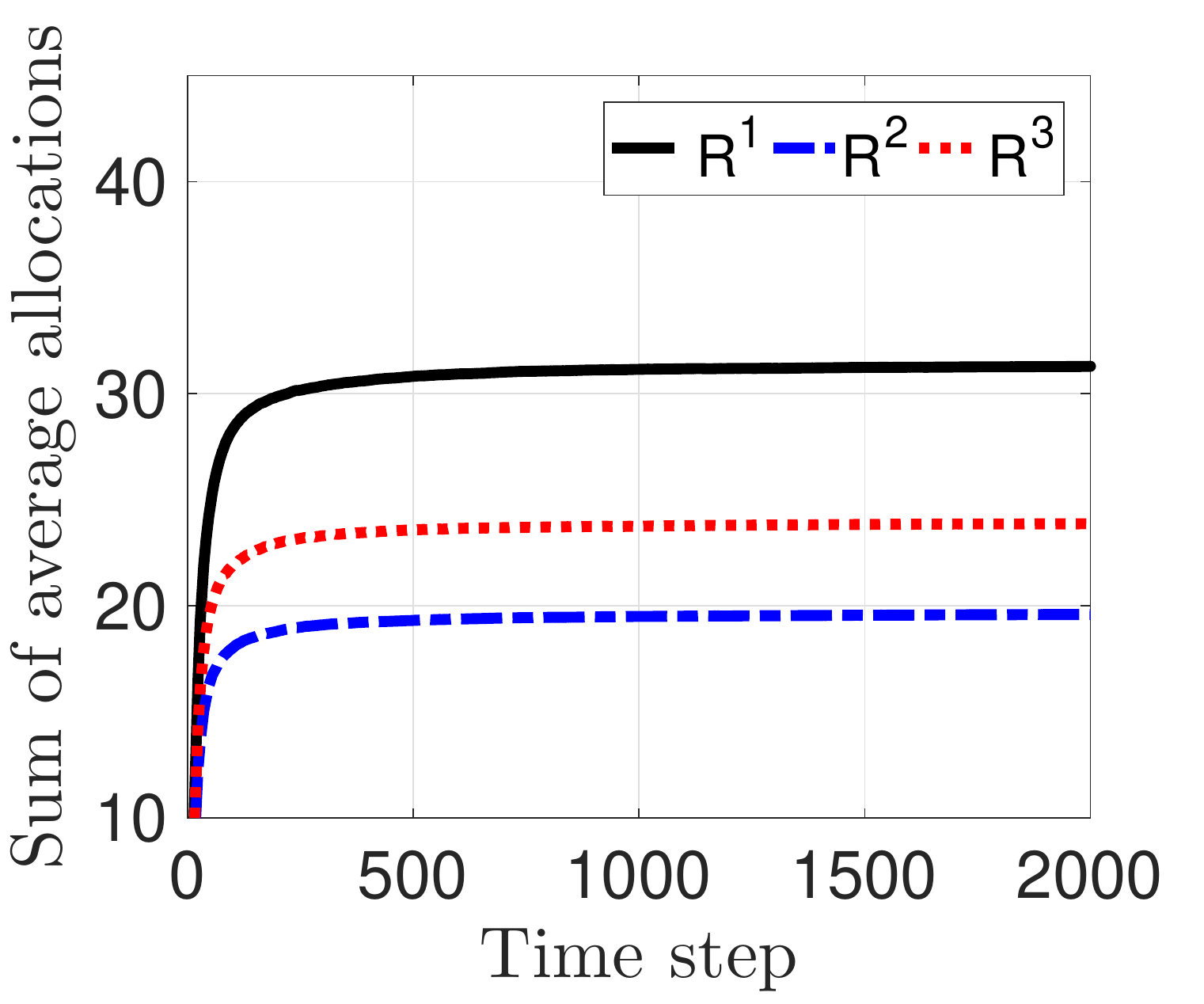}}
	\label{sum_avg}\hfill
	\subfloat[]{%
		\includegraphics[width=0.33\linewidth]{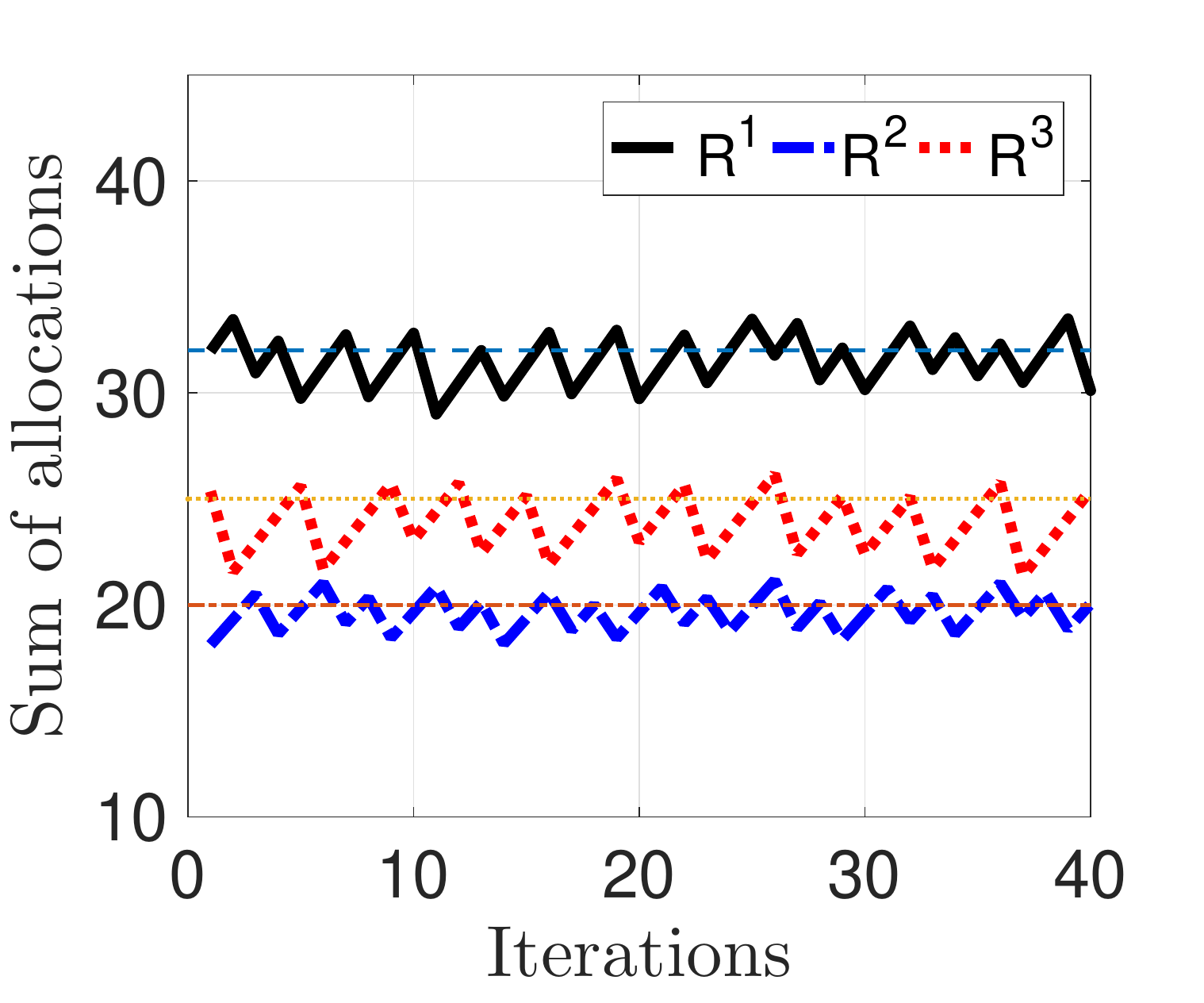}}
	\subfloat[]{%
		\includegraphics[width=0.33\linewidth]{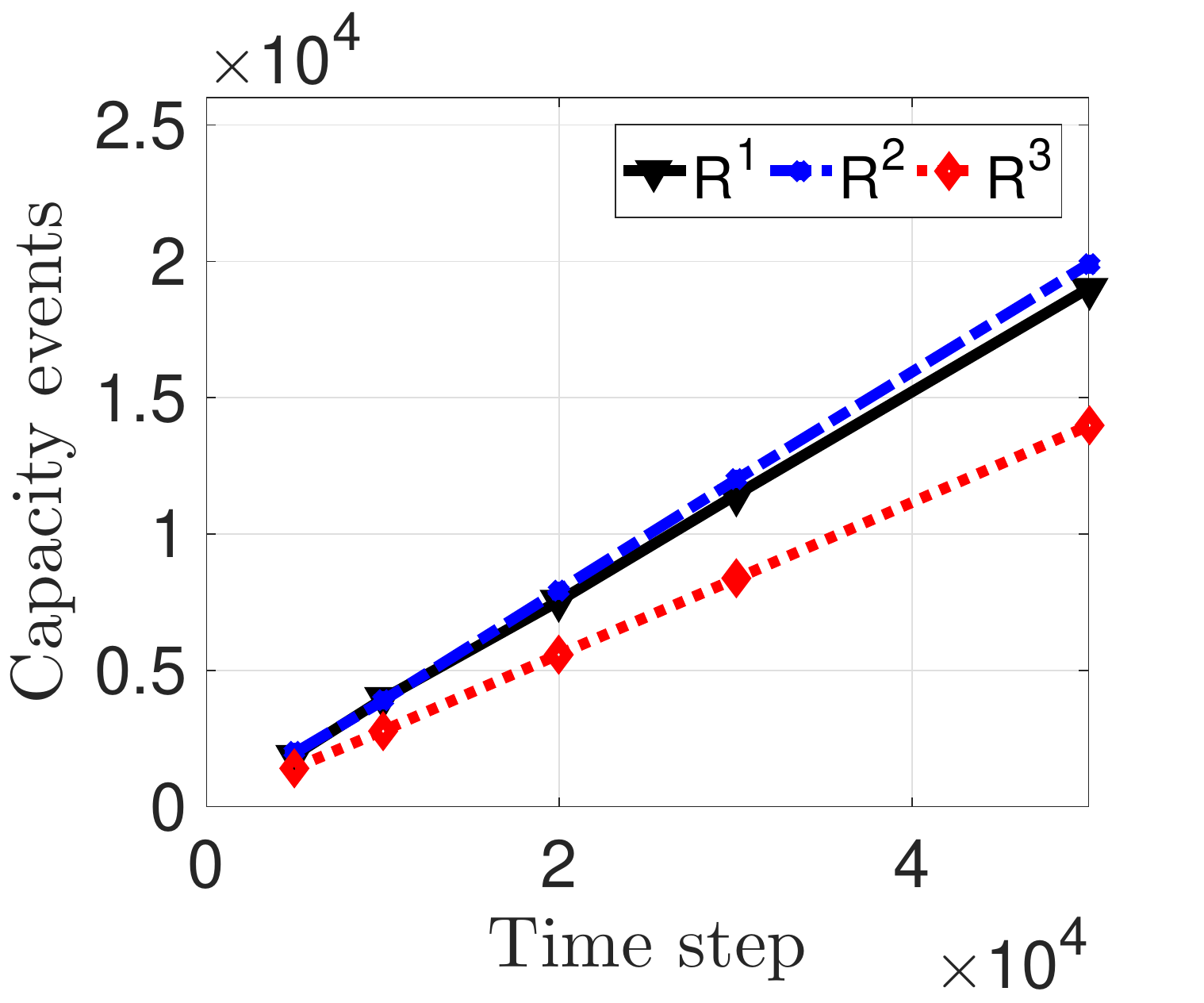}}
	\label{sum_alloc}\hfill
	\label{instant_alloc}\hfill
	\caption{(a) Evolution of sum of average allocations of resources, the sum of average allocations of a resource converges to its capacity, (b) sum of allocations of a resource is concentrated around its capacity, capacities are $C^1=32$ GB, $C^2=20$ $\mathrm{GB}^D$ and $C^3=25$ $\mathrm{Mbps}^D$, here $1 \mathrm{GB}^D = 10$ GB  and $1 \mathrm{Mbps}^D = 10$ Mbps, (c) number of capacity events for several simulations}
	\label{fig4} 
\end{figure*}

We know that, to achieve optimality, the derivatives of the cost functions of all participating cameras for a particular resource should make a consensus, which satisfies all the Karush-Kuhn-Tucker (KKT) conditions that are necessary and sufficient conditions for optimality of \eqref{obj_fn1}, as described in Section \ref{prob_form}. Figure \ref{fig1}(c) is the error bar of derivatives $\nabla_jf_i$ of cost functions $f_i$ for single simulation calculated across all cameras, for all $j$. It illustrates that the derivatives of cost functions of all cameras with respect to a particular resource concentrate more and more over time around the same value.  Hence, the long-term average allocation of resources for the stated optimization problem is optimal.

%
%
%

 For comparison purpose, we solve the optimization problem \eqref{obj_fn1} in a centralized way using the interior-point method and denote the optimal values obtained by $x_i^{*j}$, for all $i$ and $j$. We compare these optimal values with average allocation values at largest time steps in the simulation (long-term average) obtained by our proposed algorithm, we find that the results are approximately equal. Let $K$ be the largest time step used in the simulation, Figure \ref{fig2}(a) shows the evolution of {\em absolute error} which is the absolute difference of average allocation $\overline{x}_i^j(k)$ at time step $k$ and the calculated optimal allocation $x_i^{*j}$, i.e., $|\overline{x}_i^{j}(k) - x_i^{*j}|$. We observe that the absolute error approaches close to zero over time. Additionally, we calculate the {\em relative error} which we define as the ratio of sum of absolute errors and the sum of calculated optimal allocations i.e., $\frac{\sum_{i=1}^{n} |x_i^j(k) - x_i^{*j}|}{\sum_{i=1}^{n} x_i^{*j}}$. The evolution of relative error is presented in Figure \ref{fig2}(b), which decreases with time and is very low, for the described simulation it is below $5\%$. Figure \ref{fig2}(c) illustrates that the ratio of the sum of cost functions with average allocations and the sum of cost functions with optimal allocations i.e., the ratio of $\sum_{i=1}^{n} f_i(\overline{x}_i^1(K), \overline{x}_i^2(K), \overline{x}_i^3(K))$ and $\sum_{i=1}^{n} f_i(x_i^{*1}, x_i^{*2}, x_i^{*3})$ is close to $1$, which further strengthens our claim.
Furthermore, to gather information about absolute errors $|\overline{x}_i^{j}(K) - x_i^{*j}|$ of all cameras at time step $K$, we present their histograms in Figure \ref{fig3}, we observe that the absolute error of most of the cameras are close to zero.

Figure \ref{fig4}(a) illustrates the sum of average allocations $\sum_{i=1}^{n} \overline{x}_i^j(k)$ over time. We observe that the sum of average allocations at largest time step $K$ is approximately equal to the respective capacity i.e., $\sum_{i=1}^{n} \overline{x}_i^j(K) \approxeq C^j$, for all $j$ (capacities are $C^1 = 32$ GB, $C^2 = 20$ $\mathrm{GB}^D$ and $C^3 = 25$ $\mathrm{Mbps}^D$). Figure \ref{fig4}(b) shows the sum of instantaneous allocations $\sum_{i=1}^{n} x_i^j(k)$ of resource $R^j$ for last $40$ time steps. We observe that the sum of instantaneous allocations are concentrated around the respective capacities. To reduce the overshoots of total allocations of resource $R^j$, we assume $\gamma^j < 1$ and modify the algorithm of control unit to broadcast the capacity event signal $S^j(k)=1$ when $\sum_{i=1}^{n} x_i^j(k) > \gamma^j C^j$, for all $j$ and $k$.
%
%
 Furthermore, the number of capacity events is the communication overhead of the system to reach the consensus of derivatives of all cameras with respect to a particular resource, which is illustrated in the Figure \ref{fig4}(c) for several simulations. For example, the number of capacity events broadcast by the control unit in a simulation running for $30000$ time steps are $11427, 11988$ and $8355$, for resources $R^1,R^2$ and $R^3$, respectively, which are the communication overhead of the system in bits for the respective resource. Notice that the communication overhead is very low for each resource. It is also observed that the number of capacity events increases approximately linearly with time steps for different simulations.

\section{Conclusion} \label{conc} In this paper a distributed algorithm is proposed. The algorithm solves the multi-variate optimization problems for capacity constraint problems in a distributed manner. It is done by extending a variant of AIMD algorithm. The features of the proposed algorithm are; it involves little communication overhead, there is no inter-device communication needed and each Internet-connected device has its own private cost functions. It is shown in the paper that the long-term average allocation of resources converge to approximately same values as if the optimization problem under consideration is solved in a centralized setting.

 It is interesting to solve the following open problems: first is to provide a theoretical basis for the proof of convergence and second is to find the bounds for the rate of convergence, and its relationship with different parameters or the number of occurrence of capacity events. The work can also be extended in several application areas like Cloud computing, smart grids or wireless sensor networks, where sensors have very limited processing power and battery life.

\section{Acknowledgment}
 The work is supported partly by Natural Sciences and Engineering Research Council of Canada grant RGPIN-2018-05096 and by Science Foundation Ireland grant 16/IA/4610.
\bibliographystyle{IEEEtran}
\bibliography{DistOpt_bib} 
\end{document}